\documentclass[12pt]{article}
\usepackage{xcolor}
\usepackage[margin=1in]{geometry}
\usepackage{amsmath,amssymb,microtype}
\usepackage{thmtools,thm-restate}
\usepackage{amsthm}
\usepackage[colorlinks=true,allcolors=blue,pagebackref]{hyperref}
\renewcommand*{\backref}[1]{}
\renewcommand*{\backrefalt}[4]{%
  \ifcase#1\relax
  \or (cit. on page~#2.)%
  \else (cit. on pages~#2.)%
  \fi
}
\usepackage[capitalize,noabbrev]{cleveref}
\usepackage{mathtools}
\usepackage{enumerate}

\newtheorem{theorem}{Theorem}[section]
\newtheorem{lemma}[theorem]{Lemma}
\newtheorem{corollary}[theorem]{Corollary}
\theoremstyle{remark}
\newtheorem*{remark}{Remark}
\newcommand{\F}{\mathbb F}
\newcommand{\E}{\mathop{\mathbb E\/}}
\newcommand{\zo}{\{0,1\}}
\newcommand{\XMajkl}{\mathsf{Maj}_{\ell}^{\oplus k}}

\newcommand{\polylog}{\mathsf{polylog}}
\newcommand{\poly}{\mathsf{poly}}
\newcommand{\eps}{\varepsilon}
\newcommand{\NP}{\mathsf{NP}}

\newcommand{\corr}{\mathsf{corr}}
\newcommand{\Parity}{\mathsf{Parity}}
\newcommand{\High}{\mathsf{High}}

\newif\ifComments
\Commentstrue

\ifComments    
\newcommand{\manu}[1]{\textcolor{purple}{(\textbf{Manu:} #1)}}
\newcommand{\avishay}[1]{\textcolor{blue}{(\textbf{Avishay:} #1)}}
\newcommand{\shachar}[1]{\textcolor{red}{(\textbf{Shachar:} #1)}}
\newcommand{\pooya}[1]{\textcolor{teal}{(\textbf{Pooya:} #1)}}
\newcommand{\eshan}[1]{\textcolor{brown}{(\textbf{Eshan:} #1)}}
\newcommand{\chin}[1]{\textcolor{orange}{(\textbf{Chin:} #1)}}
\else
\newcommand{\manu}[1]{}
\newcommand{\avishay}[1]{}
\newcommand{\shachar}[1]{}
\newcommand{\pooya}[1]{}
\newcommand{\eshan}[1]{}
\newcommand{\chin}[1]{}
\fi

\theoremstyle{definition} %
\newtheorem{definition}{Definition} %

\newtheorem{maintheorem}{Theorem}
\crefname{maintheorem}{Theorem}{Theorems}
\Crefname{maintheorem}{Theorem}{Theorems}

\title{Exponential Correlation Bounds for Polynomials}
\author{%
  \begin{tabular}{@{}c@{\qquad}c@{\qquad}c@{}}
    \normalsize Eshan Chattopadhyay%
    \thanks{\texttt{eshan@cs.cornell.edu}.
      Supported by NSF Award CCF-2514586.}
    &
    \normalsize Pooya Hatami%
    \thanks{\texttt{hatami.2@osu.edu}.}
    &
    \normalsize Chin Ho Lee%
    \thanks{\texttt{chinho.lee@ncsu.edu}.}
    \\[3pt]
    \small Cornell University
    & \small Ohio State University
    & \small NC State University
    \\[1.5em]
    \normalsize Shachar Lovett%
    \thanks{\texttt{slovett@ucsd.edu}.
      Supported by Simons Investigator Award \#929894
      and NSF Award CCF-2425349.}
    &
    \normalsize Avishay Tal%
    \thanks{\texttt{atal@berkeley.edu}.
    Supported by NSF CAREER award CCF-2145474.}
    &
    \normalsize Emanuele Viola%
    \thanks{\texttt{mathematicsoftheimpossible@gmail.com}. Supported by NSF award 2430026.}
    \\[3pt]
    \small UC San Diego
    & \small UC Berkeley
    & \small Northeastern University
  \end{tabular}%
}

\date{\today}
\hypersetup{pdftitle={Exponential correlation bounds for polynomials}}
\begin{document}
\maketitle
\begin{abstract}   
We prove that the XOR of $k$ majorities on disjoint blocks of $\ell$ bits has correlation at most $(4d/\sqrt{\ell})^k$ with every Boolean-valued degree-$d$ polynomial over an arbitrary  field. By known techniques, we obtain explicit pseudorandom generators that fool such polynomials on $n$ Boolean inputs with seed length $d^2\cdot\polylog(n/\varepsilon)$.  The correlation and the generator extend to small $\mathsf{AC}^0[q]$ circuits, for every fixed prime power $q$.  These results resolve several long-standing questions in circuit complexity and pseudorandomness.

\end{abstract}
\section{Introduction}\label{sec:intro}
Proving correlation bounds for low-degree polynomials is a fundamental
challenge in theoretical
computer science. Such bounds have received a lot of attention,
because they are a prerequisite
for progress on a number of other long-standing challenges in
computational complexity,
related to matrix rigidity, multiparty communication complexity,
circuit lower bounds,
pseudorandomness, and more. For background on correlation bounds, see the survey \cite{corr-survey} or the book \cite{moti}.

For two Boolean functions $f,g\colon \zo^n \rightarrow \zo$, we define their correlation by 
$$\corr(f,g)\coloneqq |\E_x[(-1)^{f(x)+g(x)}]|,$$ where $x$ is uniform over $\zo^n$. 

For degree-$d$ polynomials over $\F_2$, the fundamental works of Razborov \cite{Raz87} and Smolensky \cite{Smo87, Smolensky93} proved correlation $O(d/\sqrt{n})$ with the Majority function on $n$ bits, and this bound is tight~\cite{Viola-L12requiresCor}. Remarkably, for polynomials with degree $d\ge \log n$, no explicit function (e.g. in $\NP$) was known to have asymptotically smaller correlation. Improving this bound was a prerequisite for each of the other long-standing challenges mentioned above.

Here, we finally solve this challenge and in fact obtain exponentially small correlation bounds, $\exp(-\Omega(n/d^2))$. Our results hold over every finite field, and in fact every field. To keep this exposition focused, we consider Boolean-valued polynomials over a finite field $\F$, meaning polynomials that output $0$ or $1$ on every input in $\zo^n$. We bound their correlation with the XOR of Majority on $k$ disjoint blocks of $\ell$ bits, for odd $\ell$, denoted $\XMajkl$. This is a natural candidate that was considered in previous works \cite{ViolaGF2,buhrman2007computation,VW08,chattopadhyay2020xor,Viola-L12requiresCor}.  

\begin{maintheorem}\label{thm:total-degree}
For every field $\F$ and every Boolean-valued polynomial $P \colon\zo^{k\ell}\to\zo$ over $\F$
of  total degree at most $d$,  we have
\[
 \corr(\XMajkl,P) \le\left(\frac{4d}{\sqrt\ell}\right)^k.
\]
\end{maintheorem}

Over $\F_2$, this bound is tight up to an absolute constant in the base. This follows by taking the XOR of degree-$d$ approximating polynomials for Majority constructed in \cite{Viola-L12requiresCor} on the $k$ disjoint blocks.\footnote{An earlier version of this paper was restricted to $\F_2$ and used a different weight function, yielding a slightly better constant in the base. That proof remains valid. We changed the weight function (see \Cref{subsec:weights}) to extend the argument to arbitrary fields.}

Our general result also allows Majority to be replaced by any Boolean-valued function that is \emph{versatile} over $\F$, such as Parity when the characteristic of $\F$ is different from two. For arbitrary $\F$-valued polynomials over a finite field $\F$, we obtain analogous bounds after applying any Boolean map to their outputs, with an additional dependence on $|\F|$. Over a prime field $\F_p$, we also obtain bounds for correlations with complex phases $\omega^{P(x)}$, where $\omega=e^{2\pi i/p}$, with an additional dependence on $p$. We give the precise statements and compare these extensions with prior work in \Cref{subsec:main-corollaries}.

\paragraph{Applications.}
\Cref{thm:total-degree} has several applications. First, we note that taking $\ell=64 d^2+1$,
gives an explicit function on $n = \ell k$ variables (namely $\XMajkl$) with correlation at most $2^{-k}=2^{-\Omega(n/d^2)}$ with all Boolean-valued degree-$d$ polynomials over any field.

These bounds also have consequences for circuit lower bounds.
For any fixed prime power $q=p^r$, the polynomial approximation
method of Razborov and
Smolensky~\cite{Raz87,Smolensky93,DBLP:journals/toc/Kopparty018}
simulates constant-depth circuits with $\mathrm{MOD}_q$
gates,\footnote{A $\mathrm{MOD}_q$ gate outputs $1$ if the sum of
its input bits is nonzero modulo $q$, and $0$ otherwise. In
particular, $\mathrm{MOD}_2$ is a parity gate.}
namely $\mathsf{AC}^0[q]$ circuits, by low-degree polynomials
over $\F_p$. Our correlation bounds for polynomials thus give an explicit
function with correlation at most $2^{-n^{\Omega(1)}}$ with
$\mathsf{AC}^0[q]$ circuits of size $2^{n^{\Omega(1)}}$
at any fixed depth. By the discriminator
lemma~\cite{HMPST93}, these correlation bounds also imply
$2^{n^{\Omega(1)}}$ size lower bounds for circuits consisting
of a majority gate applied to $\mathsf{AC}^0[q]$ circuits.

Next, we discuss applications to derandomization, another central goal in complexity theory. The Nisan-Wigderson generator \cite{Nisan1994hardness}, based on our hard function, gives an explicit pseudorandom generator (PRG) for Boolean-valued degree-$d$ polynomials over any field on $n$ Boolean inputs with seed length $\poly(d,\log(n/\eps))$. Applying the polarizing random walk generator \cite{DBLP:journals/toc/ChattopadhyayHH19,ChattopadhyayGLLS21} gives an improved seed length of $d^2\cdot \polylog(n/\eps)$. Over $\F_2$, all previous constructions \cite{Bog05,Lov09,Viola-d} had an exponential dependence on the degree, and so gave nothing for $d\ge \log n$.
As part of the analysis of the polarizing random walk generator, we resolve a conjecture of \cite{ChattopadhyayGLLS21} on the growth of Fourier coefficients of low-degree polynomials over $\F_2$. In fact, we show that for any Boolean-valued degree-$d$ polynomial $P$ over \emph{any field}, the function $(-1)^P$ has level-$k$ Fourier coefficients whose sum has absolute value bounded by $O(d)^k$ for all $k$. We provide details in \Cref{sec:fourier,subsec:fourier_prg}.

Using the same polynomial simulation as before, we derive PRGs for
$\mathsf{AC}^0[q]$ circuits for every fixed prime power $q$.
We obtain a PRG with polylogarithmic seed length for
polynomial-size circuits. The previous best seed length was
$n(1-o(1))$ \cite{FeffermanSUV10}. We provide details in
\Cref{sec:AC0_parity}.

Our generator for low-degree polynomials also yields PRGs for Boolean-valued sparse polynomials over any field. We obtain seed length $\polylog(nS/\eps)$ for polynomials with at most $S$ monomials. Over any finite field, when $S\ge n$ and $\eps\ge 1/\poly(S)$, this improves the $2^{O(\sqrt{\log S})}$ seed length of previous constructions~\cite{LVW93,Viola-SYMAC0,DBLP:conf/approx/ServedioT18,ServedioTan2022,Kumar2025} to $\polylog(S)$. 

Finally, our PRG for $\mathsf{AC}^0[2]$ yields unconditional
pseudorandom functions (PRFs) that are $\eps$-secure against $\mathsf{AC}^0[2]$
oracle circuits.  For any constant-depth, size-$s$ oracle circuits making $n$-bit queries, our PRF has key length $\polylog(ns/\eps)$. No such result was known, even for 
ordinary $\mathsf{AC}^0$ circuits.
The construction combines hashing with
a pseudorandom table and relies on our PRG for circuits with
parity gates. We provide details in \Cref{sec:prf}.

\paragraph{AI methodology.}
This work began with several authors prompting AI tools about correlation bounds for the XOR of Majority. One of these interactions eventually produced a new correlation bound, albeit much weaker than the one presented in this paper, and with a much more complicated proof.
The authors then undertook a research project to improve the bounds and simplify the argument through repeated interactions with AI tools, interspersed with their own insights and simplifications. Some intermediate proofs were extremely complex; the final proof emerged through many rounds of refinement.
The applications were obtained without AI assistance.

\paragraph{Acknowledgment.}
This work was done while most of the authors were at the Simons Institute for Theory of Computing, attending the Fall 2026 Pseudorandomness \& High-Dimensional Expansion program, which provided an ideal working environment. We thank Gil Cohen for helpful discussions, and Eric Allender and Ronald de Wolf for useful feedback.

\paragraph{Paper organization.}
We prove \Cref{thm:total-degree} in \Cref{sec:proof_main} and discuss applications in \Cref{sec:applications}. We conclude with open questions in \cref{sec:open}.

\section{Proof Overview}\label{sec:overview}
In this section, we provide an overview of the proof of \Cref{thm:total-degree}.

Let $\F$ be a field. Let $n=k\ell$, and partition the input into $k$ blocks of
$\ell$ bits each, where $\ell=2m+1$. Throughout, inputs are Boolean strings, and polynomials are evaluated on $\zo^n$. Let $h_1,\dots,h_k$ denote the
majority function on the respective blocks, and let $E\subseteq \zo^n$ be the set of
inputs on which a polynomial $P\colon\zo^n\to \F$ (not necessarily Boolean-valued) of total degree at most $d$ agrees with $H \coloneqq h_1\oplus\cdots\oplus h_k$.

Our main goal is to show that $|E|$ is not much larger than
$2^{n-1}$. Since the space of functions $f\colon E\to\F$ has dimension
$|E|$, we do this by showing that every function on $E$ is the
restriction of a function from a suitably chosen low-dimensional
space.

\paragraph{Smolensky in the language of bases.}
We first review Smolensky's correlation bound~\cite{Smo87,Smolensky93} for $k=1$ majority, phrased in the language we use later for multiple blocks. Let $h$ denote the majority function on $\ell$ bits, and let $E$ be the set of inputs on which a polynomial $P\colon\zo^\ell\to \F$ of degree at most $d$ agrees with $h$. The vector space of functions $f\colon E\to\F$ has dimension exactly $|E|$, so to show that $|E|$ is not much larger than $2^{\ell-1}$ it suffices to bound the dimension of this space. We do so by picking a basis for the space of functions on the full cube, and then showing that only slightly more than half of the basis elements are needed once the domain is restricted to $E$.

The basis comes from the interpolation underlying Smolensky's argument. Let $H_0=\{x\in\zo^\ell:|x|\le m\}$ and $H_1=\{x\in\zo^\ell:|x|\ge m+1\}$ be the two majority halves, so that $h$ is zero on $H_0$ and one on $H_1$, and let $U_{\le m}$ be the space of multilinear polynomials of degree at most $m$. The monomials $x^A$ with $|A|\le m$ form a basis of $U_{\le m}$, and
\[
\dim(U_{\le m})=\sum_{j=0}^m\binom{\ell}{j}=2^{\ell-1}=|H_0|=|H_1|.
\]
The interpolation fact is that an arbitrary assignment of values on either $H_0$ or $H_1$ extends uniquely to a polynomial in $U_{\le m}$. For $H_0$ this can be seen directly: if $u(x) =\sum_{|A|\le m}c_Ax^A$, its value at the input, $\mathbf{1}_B$, whose set of ones is $B$ is $\sum_{A\subseteq B}c_A$. 
To interpolate the coefficients $c_A$ from the evaluations $u(\mathbf{1}_B)$ for $|B|\le m$, start with $B = \varnothing$ from which $c_{\varnothing}$ is implied, and proceed by increasing $|B|$ and retrieving $c_{B}$ from all previous $c_{A}$'s and $u(\mathbf{1}_B)$.
We get that the prescribed values on $H_0$ determine every coefficient $c_A$ for $|A|\le m$ uniquely. Complementing all input bits exchanges the two halves and preserves the degree bound, proving the assertion for $H_1$ as well.

Consequently, every function on the cube has a unique representation $q+hr$, where $q,r\in U_{\le m}$: choose $q$ to interpolate the function on $H_0$, and $r$ to interpolate the difference on $H_1$. This gives what we call the \emph{$h$-basis}, consisting of elements $x^A$ and $x^Ah$ with $|A|\le m$. The key additional insight is that, on $E$, we have $x^Ah=x^AP$, whose degree is at most $m$ when $|A|\le m-d$. Thus we can discard these elements, retaining only the ordinary monomials of degree at most $m$ and the elements $x^Ah$ with $m-d<|A|\le m$. Counting the retained elements gives $|E|\le 2^{\ell-1}+O(d2^{\ell}/\sqrt\ell)$. For a Boolean valued $P$,  applying the same argument to $1-P$ bounds the complement of $E$. This is the correlation bound $O(d/\sqrt\ell)$ of Smolensky~\cite{Smo87}; our contribution is to make a similar basis argument work across $k$ blocks.

We note that the above is a slight reformulation of the usual presentation of Smolensky's argument, where one writes a function on $E$ as $q+hr=q+Pr$, with $q,r\in U_{\le m}$, and bounds $|E|$ by the dimension of the space of polynomials of degree at most $m+d$. In the $h$-basis, we instead identify explicitly which basis elements are redundant: if $|A|\le m-d$, then on $E$, $x^Ah=x^AP$, and $x^AP$ has degree at most $m$, so it lies in $U_{\le m}$. Thus only the \emph{high} elements $x^Ah$ with $|A|>m-d$ need to be retained. This gives the same dimension bound as the usual degree-counting argument, but the new high/low formulation is what extends naturally to several blocks.

\paragraph{The argument for several blocks.}
It is natural to consider the products of the single-block bases. Write $\mathbf A=(A_1,\dots,A_k)$, where $A_i$ is a set of coordinates from block $i$, and let $x^{\mathbf A}\coloneqq \prod_{i=1}^kx^{A_i}$. The resulting basis, which we call the $h$-basis, consists of $x^{\mathbf A}h^S$, where $|A_i|\le m$ for every $i$, $S\subseteq[k]$, and $h^S=\prod_{i\in S}h_i$, where $h_i$ is a copy of $h$ applied to the $i$-th block. Thus the $h$-basis is obtained by applying Smolensky's interpolation separately to each block, and then taking products across the blocks.

The single-block argument suggests that when $x^{\mathbf A}$ has low degree in some block $j$, we should try to eliminate the factor $h_j$ from $h$-basis elements $x^{\mathbf A}h^S$ with $j\in S$. In the single-block case, this was done simply by replacing $h_j$ with $P$, since they agree on $E$. In the general case, we rely on the fact that on $E$, for every $j$, 
\[
h_j = P\oplus \left( \bigoplus_{i\neq j} h_i\right). 
\]
Here, $\oplus$ is the binary XOR regardless of the field.
In characteristic two, this is simply the substitution 
\[
h_j=P+\sum_{i\ne j}h_i,
\] 
which rewrites a basis element containing $h_j$ as a combination of terms in which $h_j$ has disappeared. For characteristic different from two, passing from $\{0,1\}$ to $\{+1,-1\}$ turns XOR into multiplication. Expanding the resulting identity and isolating $h_j$ gives a rewriting rule of the form
\[
h_j=P+\sum_{\varnothing\neq U\subseteq[k]\setminus\{j\}}
c_U\widetilde P\,h^U,
\]
where $\widetilde P=1-2P$ has degree at most $d$ and the coefficients $c_U = (-2)^{|U|-1}$ lie in $\F$. The important point is that none of the terms on the right-hand side contain $h_j$. Thus, substituting for $h_j$ removes the selected $h_j$ factor, at
the cost of multiplying some terms by a degree-$d$ polynomial. 

To make this rewriting strategy work, we use a slightly more conservative notion of low degree than in the single-block argument. We call block $i$ \emph{low} for $\mathbf A$ if $|A_i|\le m-2d$, and \emph{high} otherwise. If every block is high for $\mathbf A$, we retain all $2^k$ basis elements $x^{\mathbf A}h^S$. Otherwise, let $j=j(\mathbf A)$ be the first low block and retain precisely those $x^{\mathbf A}h^S$ with $j\notin S$. Let $\mathcal R$ denote the span of the retained elements. We will prove that $\mathcal R$ still interpolates every function on $E$, by showing via induction that every basis element agrees on $E$ with a function in $\mathcal R$.

Consider a ``discarded'' basis element $x^{\mathbf A}h^S$, so that $\mathbf A$ has a low block and $j=j(\mathbf A)\in S$. Then, substitute for $h_j$ using the rewriting rule above. Each resulting term has the form
$x^{\mathbf A}h^TQ$, where $j\notin T$ and $Q\in\{P,\widetilde P\}$ has degree at most $d$. After multiplying by $Q$ and expanding, the resulting terms may have some blocks with degree larger than $m$. Whenever this happens, we use interpolation in that block to rewrite the corresponding monomial as $q_i+h_i r_i$, with $\deg(q_i),\deg(r_i)\le m$, simplify using $h_i^2=h_i$, and continue until the expression is written in the $h$-basis.

Crucially, this rewriting is never needed in the selected block $j$. Since $|A_j|\le m-2d$, multiplication by $Q$ raises the ordinary degree in block $j$ to at most $m-2d+d\le m$. 
Hence no $h_j$ factor can be introduced in block $j$. In particular, if a resulting basis element is still not retained, its selected low block must differ from $j$, and therefore some block must have changed its low/high status.

It remains to show that this rewriting terminates. Let
\[
\High(\mathbf A)\coloneqq\{i:|A_i|>m-2d\},
\]
and assign to each basis element the weight
\[
w(x^{\mathbf A}h^S)\coloneqq 2|\mathbf A|-(2d+1)|\High(\mathbf A)|,
\]
where $|\mathbf A| = \sum_{i}|A_i|$. Note that the weight depends only on $\mathbf A$. Multiplying by a degree-$d$ polynomial raises the total $x$-degree by at most $d$, and hence the weight by at most $2d$, unless some block changes its low/high status. The weight is designed so that any such change saves an additional $2d+1$. The low-to-high direction is immediate. In the other direction, a high-to-low change is possible only if multiplication
first pushes the degree in that block above $m$, after which
basis reduction brings it down to at most $m-2d$. Thus the
degree drops from at least $m+1$ to at most $m-2d$, a saving
of at least $2(2d+1)$ in weight, while losing the high-block
discount costs only $2d+1$.

Since every non-retained term produced by the rewriting has a changed set of high blocks, its weight is at most
\[
w(x^{\mathbf A}h^S)+2d-(2d+1)<w(x^{\mathbf A}h^S).
\]
Induction on the weight therefore shows that every basis element agrees on $E$ with a function in $\mathcal R$, and so $\mathcal R$ interpolates every function on $E$.

It remains to count the retained elements. For every $\mathbf A$ having at least one low block, we retain exactly half of the $2^k$ choices of $S$; we retain all choices only when every block is high. The fraction of one-block monomials with degree greater than $m-2d$ is at most $4d/\sqrt{\ell}$, so the fraction of all-high $\mathbf A$ is at most
\[
\left(\frac{4d}{\sqrt\ell}\right)^k.
\]
It follows that
\[
\frac{|E|}{2^n}
\le
\frac{\dim\mathcal R}{2^n}
\le
\frac12\left(1+\left(\frac{4d}{\sqrt\ell}\right)^k\right).
\]
This agreement bound holds for an arbitrary polynomial $P$ over $\F$ of degree at most $d$.

\paragraph{Correlation over arbitrary fields.}
If $P$ is Boolean-valued, then $1-P$ agrees with $H = h_1 \oplus\dots\oplus h_k$ exactly on the complement of $E$. Applying the agreement bound to both $P$ and $1-P$ therefore gives $\corr(H,P)\le(4d/\sqrt\ell)^k$.%

\paragraph{Correlation with a complex phase.}
For a prime field $\F_p$, we also consider correlation with the complex phase $\omega^{P(x)}$, where $\omega=e^{2\pi i/p}$. We show that 
\[
 \left|\E_x[(-1)^{H(x)}\omega^{P(x)}]\right| \le p \cdot (4d(p-1)/\sqrt{\ell})^k.
\]

\section{Proof of Main Result}
\label{sec:proof_main}
\label{sec:fp-proof}

We prove a general agreement bound for XORs of versatile functions over
arbitrary fields, from which \Cref{thm:total-degree} follows. We first introduce some notation and conventions used throughout this section.

For a finite set $X$, we write $x\sim X$ for sampling a uniformly random element from $X$.
Let $\F$ be an arbitrary field, and let $p$ denote its characteristic, which is either zero or a prime.
If $\F$ is finite, then $p$ is prime and $|\F|=p^r$ for some integer $r\ge 1$. 
We work throughout this section with multivariate polynomials $P\in \F[x_1, \ldots, x_n]$ evaluated on $\zo^n$. Since we are interested in their behavior as functions, such polynomials are without loss of generality multilinear using the identity $x_i^2 = x_i$ that holds over $\zo$. We identify such polynomials also as functions $P\colon\zo^n \to \F$.
We identify $\{0,1\}$ as a subset of $\F$ in the natural way. We say that a polynomial $P\colon\zo^n \to \F$ is Boolean-valued if $P(x) \in \zo$ for every $x\in \zo^n$. We also may write it as $P\colon \zo^n \to\zo$ when the context of the field we're working over is clear. We refer to the degree of a multivariate polynomial as its total degree.

\subsection{Versatile Functions}
Our main theorem establishes exponentially decaying correlation bounds between the XOR of $k$ copies of any versatile function on $\ell$ bits (e.g., the Majority function) and any Boolean-valued polynomial $P\colon \zo^n \to \F$ of degree at most $d$.

First, we define the concept of a versatile function. Let $\ell = 2m+1$ be an odd integer.
Consider the space $U$ of multilinear polynomials on $\ell$ variables over $\F$. Its dimension is $2^{\ell}$.
Let $U_{\le m}\subseteq U$ be the subspace consisting of polynomials of degree at most $m$. Note that 
$\dim (U_{\le m})=\sum_{i=0}^m\binom\ell i=2^{\ell-1}$,
exactly half the dimension of $U$.

\begin{definition}[Versatile Function~\cite{kopparty2011complexity}]
We say that a function $h\colon\zo^\ell \to \F$  is versatile over $\F$ if any function $f\colon\zo^\ell \to \F$ can be written (uniquely) as 
\begin{equation}\label{eq:one-block}f = q + hr, \qquad \text{for~} q, r \in U_{\le m} .
\end{equation} 
\end{definition}

We remark that uniqueness comes for free from dimension counting, as $2\dim(U_{\le m}) = \dim(U)$.
The versatility of a \emph{Boolean-valued} $h$ is equivalent to an interpolating property with respect to $h^{-1}(0)$ and $h^{-1}(1)$. In particular, it implies that $h$ is balanced, i.e., $|h^{-1}(0)| = |h^{-1}(1)| = 2^{\ell-1}$.

We record the following useful result and supply its proof for completeness.
\begin{restatable}[Versatility in terms of Interpolation]{lemma}{versatileinterpolation}\label{lemma:versatility_vs_interpolation}
A function $h\colon\zo^\ell \to \zo$ is versatile over $\F$ if and only if both $H_0 = h^{-1}(0)$ and $H_1 = h^{-1}(1)$ are interpolating sets for $U_{\le m}$.
\end{restatable}
\begin{proof} 
In the forward direction, let $g\colon H_0 \to \F$. Extend $g$ to a function $g'$ from $\zo^\ell \to \F$ by mapping $x\notin H_0$ arbitrarily to values in $\F$.
Then, by versatility of $h$, $g' = q+hr$ on $\zo^\ell$ for $q,r\in U_{\le m}$. Thus, on $x\in H_0$ we have $g(x) = g'(x) = q(x)$ and so $q$ interpolates $g$ on $H_0$.

Similarly, for any $g\colon H_1 \to \F$,  extend $g$ to a function $g'$ from $\zo^\ell \to \F$ by mapping $x\notin H_1$ arbitrarily to values in $\F$.
Then, by versatility of $h$, $g' = q+hr$ on $\zo^\ell$ for $q,r\in U_{\le m}$.
Thus, on $x\in H_1$ we have $g(x) = g'(x) = q(x) + r(x)$ and so $q+r\in U_{\le m}$ interpolates $g$ on $H_1$. Thus, $|H_0|$ and $|H_1|$ are bounded by the dimension of $U_{\le m}$ which is $2^{\ell-1}$. Combined with $|H_0|+|H_1|=2^\ell$, this implies $|H_0|=|H_1|=2^{\ell-1}$, which guarantees uniqueness of the interpolant.

In the other direction, assume $H_0, H_1$ are interpolating sets for $U_{\le m}$. Then, both sets have size equal to $\dim(U_{\le m}) = 2^{\ell-1}$, as otherwise, one of them would have size larger than $\dim(U_{\le m})$ which contradicts being an interpolating set.

Interpolate $f$ on $H_0$ and $H_1$ respectively to get polynomials $q_0, q_1 \in U_{\le m}$
with $f|_{H_0} \equiv q_0$ and $f|_{H_1} \equiv q_1$. We claim that 
\[
f = q_0 + h(q_1 - q_0).
\]
This is immediate from the cases $h(x)=1$ and $h(x)=0$ on $H_1$ and $H_0$, respectively.
Uniqueness follows from dimension counting.
\end{proof}

\medskip
For our applications, Majority is versatile over any field $\F$ by the interpolation argument in
\cref{sec:overview}. If $p\ne 2$, then $\Parity_\ell$ (parity of $\ell$ bits) is versatile over $\F$ as well.

\begin{lemma}
    Let $\F$ be a field of characteristic $p\ne 2$. Then $\Parity_\ell$ is versatile over $\F$.
\end{lemma}
\begin{proof}
As $\Parity_{\ell}$ is Boolean-valued, it suffices to prove the interpolation property.
Write $y_i=1-2x_i=(-1)^{x_i}$. Since $p\ne 2$, the two functions
$1$ and $y_i$ form a basis for the functions of the $i$-th input bit.
Taking tensor products over the $\ell$ input bits shows that the
characters $y^S=\prod_{i\in S}y_i$, for $S\subseteq[\ell]$, form a
basis for $U$.

Next, the key is the identity
\[
y^A=y^{[\ell]} \cdot y^{[\ell]\setminus A},
\]
where $y^{[\ell]} = (-1)^{\sum_{i=1}^{\ell} x_i}$.
On $\Parity_\ell^{-1}(b)$, the factor $y^{[\ell]}$ is simply the constant
$(-1)^b$. Since $\ell=2m+1$, one of $A$ and $[\ell]\setminus A$ has
size at most $m$, and hence on
$\Parity_\ell^{-1}(b)$ every character is a fixed scalar multiple of a character of degree at most $m$. There are exactly
$2^{\ell-1} = |\Parity_\ell^{-1}(b)|$ such characters, so the
interpolation is unique.\end{proof}

\subsection{Statement of the Main Result and Its Corollaries}
\label{subsec:main-corollaries}

\begin{theorem}\label{thm:versatile-agreement}
Let $\F$ be a field, and let $k\ge1$ and $m\ge0$. Set $\ell=2m+1$ and $n=k\ell$.
Let $h\colon\zo^{\ell}\to \zo$ be a Boolean-valued versatile function over $\F$.
Let $H\colon\zo^{n} \to \zo$ be 
\[
  H(x) = \bigoplus_{i=1}^k h_i(x),
\]
where $x = (x^{(1)}, \ldots, x^{(k)})$, each block
of $x$, denoted $x^{(i)}$, consists of $\ell$ bits, and $h_i(x) = h(x^{(i)})$.

Then, for every integer $d\ge0$ and polynomial
$P\colon\zo^{n}\to\F$ of total degree at most $d$,
\[
 \Pr_{x\sim \zo^n}[P(x)=H(x)]
 \le \frac12+\frac12\left(\sqrt{\frac{2}{\pi}}\cdot\frac{4d}{\sqrt\ell}\right)^k  \le \frac12+\frac12\left(\frac{4d}{\sqrt\ell}\right)^k.
\]
\end{theorem}

\begin{corollary}\label{main_cor}
Under the same conditions as \Cref{thm:versatile-agreement},
\begin{enumerate}
\item For any Boolean-valued polynomial $P\colon\zo^n \to \zo$ over $\F$ of total degree at most $d$,\looseness=-1
\[
 \corr(H,P)
 =\left|\E_{x\sim\zo^n}[(-1)^{H(x) + P(x)}]\right|
 \le \left(\sqrt{\frac{2}{\pi}}\cdot \frac{4d}{\sqrt\ell}\right)^k.
\]
\item Suppose that $\F$ is a finite field. For any function $g\colon\F \to \{0,1\}$, and any polynomial $P\colon\zo^n \to \F$ of total degree at most $d$,
\[
 \corr(H,g\circ P)
 =\left|\E_{x\sim\zo^n}[(-1)^{H(x) + g(P(x))}]\right|
 \le \left(\sqrt{\frac{2}{\pi}}\cdot \frac{4d(|\F|-1)}{\sqrt\ell}\right)^k.
\]

\item Suppose that $\F=\F_p$ is a prime field, and let $\omega=e^{2\pi i/p}$. For any polynomial $P\colon\zo^n \to \F_p$ of total degree at most $d$,
\[
 \left|\E_{x\sim \zo^n}[\omega^{P(x)} \cdot  (-1)^{H(x)}]\right| \le \frac{p}{2}\left(\sqrt{\frac{2}{\pi}}\cdot \frac{4(p-1)d}{\sqrt\ell}\right)^k.
\]

\end{enumerate}
\end{corollary}
We now prove the corollary assuming the correctness of \Cref{thm:versatile-agreement}.
\begin{proof}
For an integer $t\ge 0$, write
 $\delta_t\coloneqq
 \left(\sqrt{\frac{2}{\pi}}\cdot\frac{4t}{\sqrt\ell}\right)^k$.

\begin{enumerate}
\item Let $E=\{x\in\zo^n:P(x)=H(x)\}$.
By \Cref{thm:versatile-agreement}, over $x\sim \zo^n$,
\[
 \Pr[E]\le \frac12+\frac{\delta_d}{2}.
\]
Since $P$ is Boolean-valued, $1-P$ agrees with $H$ precisely on the
complement of $E$, denoted $E^c$. Applying \Cref{thm:versatile-agreement} to $1-P$ gives
\[
 \Pr[E^c] \le \frac12+\frac{\delta_d}{2},
\]
thus $\Pr[E]\ge 1/2 - \delta_d/2$. Consequently,
\[
 \corr(H,P)=\left|2\Pr[E]-1\right|\le\delta_d.
\]

\item Since $\F$ is finite, Lagrange interpolation gives a polynomial
$\widetilde g\in\F[z]$ of degree at most $|\F|-1$ that agrees with $g$
on $\F$. Hence $g\circ P$ is represented on $\zo^n$ by the
Boolean-valued polynomial $\widetilde g\circ P$, whose degree is at
most $(|\F|-1)d$. The first item therefore gives
\[
 \corr(H,g\circ P)\le\delta_{(|\F|-1)d}.
\]

\item Versatility implies that $h$ is balanced, and hence so is $H$, the XOR of its $k$ copies on disjoint blocks. Thus
$\E_{x\sim \zo^n}[(-1)^{H(x)}]=0$. For each $a\in\F_p$, let
$g_a(t)=\mathbf 1[t=a]$. By the second item,
\[
 \left|\E_{x\sim\zo^n}\left[(-1)^{H(x)+g_a(P(x))}\right]\right|
 \le \delta_{(p-1)d}.
\]
Using $(-1)^{g_a(P(x))}=1-2\cdot \mathbf 1[P(x)=a]$ and the balance of $H$,
we obtain
\begin{equation}\label{eq:3}
 \left|\E_{x\sim\zo^n}\left[(-1)^{H(x)}\cdot \mathbf 1[P(x)=a]\right]\right|
 \le \frac{\delta_{(p-1)d}}{2}.
\end{equation}
Finally,
\begin{align*}
 \left|\E_{x\sim\zo^n}\left[\omega^{P(x)}(-1)^{H(x)}\right]\right|
 &=\left|\sum_{a\in\F_p}\omega^a
      \E_{x\sim\zo^n}\left[(-1)^{H(x)} \cdot \mathbf 1[P(x)=a]\right]\right|\\&\le
     \sum_{a\in\F_p} \left|\omega^a\right|\cdot \left|
      \E_{x\sim\zo^n}\left[(-1)^{H(x)} \cdot \mathbf 1[P(x)=a]\right]\right|\\
 &\le p \cdot \frac{\delta_{(p-1)d}}{2}\tag{since $|\omega^a|= 1$ and using \Cref{eq:3}},
\end{align*}
as claimed.\qedhere
\end{enumerate}
\end{proof}

\paragraph{Comparison with previous bounds.}
Over a fixed odd prime field $\F_p$, earlier correlation bounds for
parity against polynomial phases and Boolean maps of polynomial outputs
were obtained by Bourgain~\cite{Bou05}, Green, Roy, and
Straubing~\cite{GRS05}, and Chattopadhyay~\cite{Chatto07}. 
Chattopadhyay's bound has the form
$\exp(-\Omega(n/C^d))$, where $C=p2^{p-1}$.
For $d\ge\log_2 n$, these estimates do not improve on the classical
$O(d/\sqrt n)$ bound for Boolean-valued
polynomials~\cite{Smo87,Smolensky93}.
Taking $h=\Parity_\ell$ and $\ell=\Theta(d^2)$ in
\Cref{main_cor} gives $\exp(-\Omega(n/d^2))$,
replacing exponential dependence on $d$ by quadratic dependence.
We note that the earlier results worked for composite moduli as well.

\paragraph{Organization.} The proof of \Cref{thm:versatile-agreement} is an extension of Smolensky's celebrated linear-algebraic encoding argument. We start by introducing a suitable basis for the space of functions from $\zo^n$ to $\F$ in \cref{subsec:basis} based on the versatile function $h$. We then state our key lemma and prove the main result assuming it in \cref{subsec:main_thm_dimension}. Finally, in \cref{subsec:weights}, we prove the lemma using a potential function called \emph{weight} that decreases under repeated applications of the simplification rules.%

\subsection{The \texorpdfstring{$h$}{h}-basis}
\label{subsec:basis}

We start from a versatile function $h$ on a single block of $\ell$ bits with the property that every function $f$ on $\ell$ bits can be written as $f = q+hr$ for polynomials $q,r$ of degree at most $m$.

Generalizing to $k>1$ blocks, we consider the space of polynomials whose degree \emph{in each block} is at most $m$, along with the corresponding functions $h_1, \ldots, h_k$.
A monomial \emph{of degree at most $m$ in each block} is specified by a
tuple of subsets $\mathbf A=(A_1,\ldots,A_k)$, where
$A_i\subseteq[\ell]$ and $|A_i|\le m$. Let $\mathcal M$ be the
set of all such tuples, and write
\[
x^{\mathbf A} \coloneqq \prod_{i=1}^k\prod_{j\in A_i}x^{(i)}_j.
\]
The degree of this monomial in the $i$-th block is $|A_i|$. 
For $S\subseteq [k]$ we write $h^S$ for $\prod_{i\in S} h_i$. It is useful to think of $h_1, \ldots, h_k$ as new formal variables. This brings us to the definition of the $h$-basis:

\begin{definition}
The \emph{$h$-basis} consists of basis functions
$\{x^{\mathbf A} h^S\}_{\mathbf A, S}$ where $\mathbf A\in \mathcal M$ and $S\subseteq [k]$.
\end{definition}

Next, we prove that this is indeed a basis for the space of functions from $\zo^{k\ell}$ to $\F$.

\begin{lemma}\label[lemma]{lem:h-decomposition}
Every function $f\colon\zo^{k\ell}\to\F$ can be written uniquely as a linear combination of basis elements in the $h$-basis.
\end{lemma}
\begin{proof}
We first show that the $h$-basis spans every function.
As every function from $\zo^{k\ell}$ to $\F$ can be expressed as a multilinear polynomial over $\F$ in variables $\{x_j^{(i)}\}_{i\in [k],j\in [\ell]}$, it suffices to show that every monomial in these variables lies in the span of the $h$-basis. To do that, first express each such monomial $x^{\mathbf B}$ for $\mathbf B = (B_1, \ldots, B_k)$ with $B_i\subseteq [\ell]$ (note that $|B_i|$ can be larger than $m$) as a product of $k$ monomials, $M_1, \ldots, M_k$, where $M_i = \prod_{j\in B_i} x_{j}^{(i)}$.
Then, use \Cref{eq:one-block} to express $M_i$ as $q_i(x^{(i)}) + h_i r_i(x^{(i)})$, where $\deg(q_i), \deg(r_i)\le m$ and multiply the results to get that $x^{\mathbf B}$ is a linear combination of the basis elements in the $h$-basis.

Uniqueness follows by counting dimensions, since we have $(2^{\ell}/2)^k \cdot 2^k = 2^{\ell k}$ basis functions, which is the dimension of the space of functions on $k \ell$ bits.
\end{proof}

\subsection{Proof of the Main Theorem by Dimension Counting}
\label{subsec:main_thm_dimension}

Let $P$ be a multilinear polynomial over $\F$ %
in $k\ell$
variables of total degree at most $d$.  Define the agreement set
\[
 E \coloneqq \{x\in \zo^{k\ell}:P(x)=H(x)\}.
\]
We need to show that the relative size of $E$ is not much larger than $1/2$.  To do so, we will precisely follow Smolensky's idea and show that every function on $E$ is a (restriction of a) function from a space of dimension only slightly larger than half that of the full function space.  The new space is defined in the $h$-basis as follows. For each monomial
\(x^{\mathbf A}\) with a low-degree block, we shall pick the first such block $j$ and keep only the \(2^{k-1}\) basis terms \(\{x^{\mathbf A}h^S\}_{S\subseteq [k]\setminus\{j\}}\) not involving $h_j$. 
For monomials $x^{\mathbf A}$ whose
degree is high in \emph{every block}, we keep all \(2^k\) terms \(\{x^{\mathbf A}h^S\}_{S\subseteq [k]}\).
A simple calculation bounds the dimension of the space spanned by these basis terms, which then bounds $|E|$. 

More precisely, we define the set of \emph{top} monomial indices
\[
 \mathcal{T}
 \coloneqq \{\mathbf A\in\mathcal M:
       |A_i|>m-2d\text{ for every }i\in[k]\}.
\]
For every non-top $\mathbf A\notin\mathcal{T}$ choose a low-degree block%
\footnote{Choosing $j(\mathbf A)$ as the first low-degree block
is only for convenience. Any fixed selection rule depending only on
the set of low-degree blocks would work.}
\[
 j(\mathbf A) \coloneqq \min\{i\in[k]:|A_i|\le m-2d\}.
\]
Define $\mathcal R$ to be the $\F$-linear span of the following basis
functions:
\begin{equation}\label{eq:R-definition}
\mathsf{Allowed} \coloneqq \{x^{\mathbf A}h^S : \mathbf A\in\mathcal{T},
                         \ S\subseteq[k]\} \,\cup\, 
 \{x^{\mathbf A}h^S : \mathbf A\notin\mathcal{T},
                         \ S\subseteq[k],
                         \ j(\mathbf A)\notin S\} .
\end{equation}
In words, for a non-top monomial $x^{\mathbf A}$ we choose one
low-degree block $j(\mathbf A)$ and forbid all products containing the
factor $h_{j(\mathbf A)}$.  For a top monomial we impose no
restriction.

Next is the key lemma enabling our dimension count.

\begin{lemma}[Agreement-Set Interpolation]\label{lem:restriction-span}
For every function $a\colon E\to\F$, there exists a function
$r\in\mathcal R$ such that $r(x)=a(x) \text{ for every }x\in E$.
\end{lemma}

Assuming this lemma, we can complete the proof of the main theorem by counting the dimension of $\mathcal R$.

\begin{proof}[Proof of \Cref{thm:versatile-agreement}]
Let
\[
 M_\ell \coloneq \sum_{r=0}^m\binom\ell r=2^{\ell-1},
 \qquad
 N_{\ell,2d} \coloneq \sum_{r=\max\{0,m-2d+1\}}^m\binom\ell r.
\]
There are $M_\ell^k$ ordinary monomials in $\mathcal M$, and exactly
$N_{\ell,2d}^k$ top monomials in $\mathcal{T}$.  For each
non-top monomial $x^{\mathbf A}$, exactly $2^{k-1}$ subsets $S\subseteq[k]$ avoid
$j(\mathbf A)$.  For each top monomial, all $2^k$ subsets are
allowed.  Hence
\begin{align}
 \dim(\mathcal R)
 &=2^{k-1}\bigl(M_\ell^k-N_{\ell,2d}^k\bigr)
     +2^kN_{\ell,2d}^k \notag\\
 &=2^{k-1}M_\ell^k+2^{k-1}N_{\ell,2d}^k \notag\\
 &=2^{k\ell-1}+2^{k-1}N_{\ell,2d}^k.
 \label{eq:R-dimension}
\end{align}

By \cref{lem:restriction-span}, the restriction map
$\mathcal R\to\F^E$ is surjective.  Since the  space of
functions on $E$ has dimension $|E|$, \Cref{eq:R-dimension} gives
\[
 |E|\le 2^{k\ell-1}+2^{k-1}N_{\ell,2d}^k.
\]
There are at most $2d$ summands in $N_{\ell,2d}$, each at most 
\(
 \binom\ell m\le \sqrt{\frac{2}{\pi}}\,\frac{2^\ell}{\sqrt\ell},
\)
which gives
 \[
 |E|\le 2^{k\ell-1}+2^{k-1}\left(2d\cdot \sqrt{\frac{2}{\pi}}\cdot \frac{2^\ell}{\sqrt\ell}\right)^k = 
 2^{k\ell}\cdot \left(\frac{1}{2} + \frac{1}{2} \cdot \left(\sqrt{\frac{2}{\pi}} \cdot \frac{4d}{\sqrt{\ell}}\right)^k\right).
 \]
Consequently,
\[
 \Pr_{x\sim \zo^n}[P(x)=H(x)] = \frac{|E|}{2^{k\ell}} \le \frac{1}{2} + \frac{1}{2} \cdot \left(\sqrt{\frac{2}{\pi}} \cdot \frac{4d}{\sqrt{\ell}}\right)^k\;.\qedhere
\]
\end{proof}

\subsection{Proof of the Agreement-Set Interpolation Lemma}
\label{subsec:weights}
It remains to prove the key \Cref{lem:restriction-span}.

We first describe the rewriting rule. 
On $E$, we have $P=H=h_1\oplus\cdots\oplus h_k$, where $\oplus$
denotes Boolean XOR. In particular, $P$ is Boolean-valued on $E$, and
for each $j\in[k]$ we can isolate
\[
 h_j=P\oplus\left(\bigoplus_{i\ne j}h_i\right).
\]

If the characteristic $p$ is $2$, this simplifies on $E$ to
$h_j=P+\sum_{i\ne j}h_i.$ If $p\ne 2$, then $2$ is invertible in $\F$, and we translate the outputs of $h_1,\ldots,h_k$ and $P$ from
$\{0,1\}$ to $\{+1,-1\}$ via the invertible map $b\mapsto1-2b$.
Under this map XOR becomes multiplication, so on $E$,
\[
 1-2h_j=(1-2P)\prod_{i\ne j}(1-2h_i).
\]
Isolating $h_j$ gives
\begin{align}
h_j
&= \frac{1-(1-2P)\prod_{i\neq j}(1-2h_i)}{2} \nonumber\\
&= \frac12-\frac{1-2P}{2}
   \sum_{U\subseteq[k]\setminus\{j\}}(-2)^{|U|}h^U \nonumber \\
&= P+\sum_{\varnothing\neq U\subseteq[k]\setminus\{j\}}
   (-2)^{|U|-1}(1-2P)h^U
   \qquad\text{on }E. \label{eq:1}
\end{align}
Crucially, none of the monomials $h^U$ on the right-hand side depends on $h_j$.
We note that the expression in \cref{eq:1} 
reduces to
$h_j=P+\sum_{i\ne j}h_i$ in characteristic $2$: when $|U|=1$ the coefficient
of $h^U$ is $1$ (using $1-2P=1$), while for $|U|\ge2$ it vanishes.

Thus, over any field $\F$, we have 
\[
h_j
= P+\sum_{\varnothing\neq U\subseteq[k]\setminus\{j\}}
   (-2)^{|U|-1}(1-2P)h^U    \qquad\text{on }E.
\]
For ease of notation, set
\[
 \widetilde P\coloneqq1-2P,
 \qquad
 c_U\coloneqq(-2)^{|U|-1}\in\F.
\]
Then the total degree of $\widetilde P$ is at most $d$, and the preceding identity becomes
\[
 h_j=P+\sum_{\varnothing\neq U\subseteq[k]\setminus\{j\}}
 c_U\widetilde P h^U
 \qquad\text{on }E.
\]

Thus, for a forbidden term $T=x^{\mathbf A}h^S$ with
$j=j(\mathbf A)\in S$, we can write
\begin{equation}
\label{eq:forbidden-rewrite}
 T=x^{\mathbf A}h^{S\setminus \{j\}}\cdot h_j
 =x^{\mathbf A}h^{S\setminus \{j\}}\cdot
 \bigg(P+\sum_{\varnothing\neq U\subseteq[k]\setminus\{j\}}
 c_U\widetilde P h^U\bigg)
 \qquad\text{on }E,
\end{equation}
where, crucially, none of the displayed $h$-monomials on the right-hand side of \eqref{eq:forbidden-rewrite} contains $h_j$.

We then express the right-hand side of \eqref{eq:forbidden-rewrite} in the $h$-basis and show that each term in its expansion is either an allowed term (in $\mathsf{Allowed}$)
or has a smaller weight, to be defined next.

\begin{definition}[$x$-Degree, High Blocks, Weight]
For $\mathbf A\in\mathcal M$, write
\[
 |\mathbf A|\coloneqq\sum_{i=1}^k|A_i|,\qquad
 \High(\mathbf A)\coloneqq\{i\in[k]:|A_i|>m-2d\}.
\]
We call block $i$ \emph{high} if $i\in\High(\mathbf A)$ and
\emph{low} otherwise.
For a basis function $x^{\mathbf A}h^S$, define its \emph{$x$-degree} as $|\mathbf{A}|$ and its \emph{block-$i$ $x$-degree} as $|A_i|$. %
Define its weight by
\[
 w(x^{\mathbf A}h^S)\coloneqq
 2|\mathbf A|-(2d+1)|\High(\mathbf A)|.
\]
\end{definition}

Crucially, the weight of $x^{\mathbf A}h^S$ depends only on $\mathbf A$, not on $S$: it is twice the
$x$-degree, with a discount of $2d+1$ for each high block. Thus,
multiplying by a polynomial in
$h_1,\dots,h_k$ and reducing via $h_i^2=h_i$ produces only basis terms of the same weight, since every resulting basis term has the same $x$-monomial
index $\mathbf A$. As we show next, multiplication by an input variable can
increase the $x$-degree by at most one, and hence the weight by at
most two. Furthermore, if a block changes its low/high status, the discount or the
degree loss from basis reduction yields an additional saving of $2d+1$. The next lemma makes this accounting precise.

\begin{lemma}[Multiplication by One Variable]\label{lem:multiplication}
Let $T=x^{\mathbf A}h^S$ be a basis function, and let
$z=x^{(i)}_e$ be an input variable in block $i$.
Suppose
$T'=x^{\mathbf A'}h^{S'}$ has nonzero coefficient in the $h$-basis expansion of $zT$. 

Then
$w(T')\le w(T)+2$. Furthermore, if block $i$ (containing $z$) changes its low/high
status, then
\[
w(T')\le w(T)+2-(2d+1).
\]
\end{lemma}

\begin{proof}
First note that only the block containing $z$ (i.e., block $i$) changes. Let $a$ and $a'$ be the
block-$i$ $x$-degrees of $T$ and $T'$, respectively.
We consider three cases:
\begin{enumerate}
    \item If $e\in A_i$, then multiplying by $z=x^{(i)}_e$ does nothing as $z^2 = z$ and thus $a'=a$.
\item If $e\notin A_i$ and $a\le m-1$,
then $zx^{\mathbf A} h^S$ is a valid basis function and thus $T' = zx^{\mathbf A} h^S$ and $a'=a+1$.
\item Finally, if $e\notin A_i$ and $a=m$, then
multiplication by $z$ increases the block-$i$ $x$-degree to
$m+1$. We then use the versatility of $h$ in the $i$-th block to
rewrite $zx^{A_i}$ as $q_i+h_i r_i$, where
$\deg(q_i),\deg(r_i)\le m$. This is called ``basis reduction''.
Each basis term $T'=x^{\mathbf A'}h^{S'}$ in this expansion has
block-$i$ $x$-degree at most $m$, and the other blocks are unaffected.
\end{enumerate}
In all cases, $a'\le a+1$.

If the low/high status of block $i$ is unchanged, $\High(\mathbf A) = \High(\mathbf A')$, so the weight
increases by at most $2$. 

If the block changes from low to high, then
necessarily $a=m-2d$, $a'=m-2d+1$, and no basis reduction occurs.
Its block-$i$ $x$-degree increases by one and it earns the discount $2d+1$, giving weight
change $2-(2d+1)$.

Finally, a high-to-low change requires basis reduction, since without
reduction the block-$i$ $x$-degree can only increase. Hence $a=m$ and
$a'\le m-2d$ and the weight change satisfies
\[
w(T')-w(T)  \le 2(a'-a)+(2d+1)\le -4d+(2d+1)=2-(2d+1).
\]
This proves both assertions.
\end{proof}

\begin{corollary}[Multiplication by a Low-Degree Polynomial]
\label{cor:polynomial-multiplication}
Let $t\ge0$ be an integer, let $T=x^{\mathbf A}h^S$ be a basis function, and let $q$ be a
multivariate polynomial in the input variables of total degree at most
$t$. Suppose that
$T'=x^{\mathbf A'}h^{S'}$ has nonzero coefficient in the $h$-basis
expansion of $qT$. Then $w(T')\le w(T)+2t$.
If $\High(\mathbf A)\neq\High(\mathbf A')$, then
\[
 w(T')\le w(T)+2t-(2d+1).
\]
\end{corollary}

\begin{proof}
By linearity, it suffices first to take $q=z_1\cdots z_s$ to be a
monomial, where $s\le t$ and each $z_r$ is an input variable. Multiply
by the variables one at a time, expanding in the $h$-basis after each
multiplication. Any basis term $T'$ in the final expansion arises
along a sequence of basis terms
\[
 T=T_0,T_1,\dots,T_s=T',
\]
where $T_r$ has nonzero coefficient in the $h$-basis expansion of
$z_rT_{r-1}$. Applying \Cref{lem:multiplication} at every step gives
$w(T')\le w(T)+2s\le w(T)+2t$.

If $\High(\mathbf A)\neq\High(\mathbf A')$, then at
least one step in the sequence changes the low/high status of a block.
Applying the stronger conclusion of \Cref{lem:multiplication} at that
step gives
\[
 w(T')\le w(T)+2s-(2d+1)\le w(T)+2t-(2d+1).
\]
The result for a general polynomial $q$ follows by expanding it into
monomials: a basis term with nonzero final coefficient must occur in
the expansion of at least one monomial of $q$.
\end{proof}

\begin{proof}[Proof of \Cref{lem:restriction-span}]
Every function on $E$ extends to the full cube, where it has an
$h$-basis expansion. It therefore suffices to show that each basis function $T=x^{\mathbf A}h^S$ agrees with some function in $\mathcal R$ on $E$. We use induction on $w(T)$, with the basis functions of minimum weight as the base case.

If $T \in \mathsf{Allowed}$ (defined in \Cref{eq:R-definition}), there is nothing to prove. Otherwise,
$\mathbf A\notin\mathcal T$ and $j=j(\mathbf A)\in S$ is a block with $|A_j|\le m-2d$.
Apply \Cref{eq:forbidden-rewrite} to get 
\[
 T=x^{\mathbf A}h^{S\setminus\{j\}}\cdot h_j
 =x^{\mathbf A}h^{S\setminus\{j\}}\cdot
 \left(P+\sum_{\varnothing\neq U\subseteq[k]\setminus\{j\}}
 c_U\widetilde P h^U\right).
\]
We show that each term on the right-hand side expands in the $h$-basis
as a linear combination of allowed terms and terms $T'$ with
$w(T')<w(T)$.

Expanding the right-hand side, and using the identities $h_i^2 = h_i$ for $i\in [k]$, each term in it is of the form
$c\,x^{\mathbf A}h^{U'}P'$, where
 $c\in\F\setminus\{0\}$, $U'\subseteq[k]\setminus\{j\}$ and
$P'\in\{P,\widetilde P\}$ has total degree at most $d$.

Fix such a term $c\,x^{\mathbf A}h^{U'}P'$, where $j\notin U'$, and let $T'=x^{\mathbf A'}h^{S'}$ be a basis function occurring in the
$h$-basis expansion of $c\,x^{\mathbf A}h^{U'}P'$.
We claim that $j\notin S'$.
Indeed, after multiplication of $c\,x^{\mathbf A}h^{U'}$ by any monomial of $P'$, its block-$j$ $x$-degree is at most
\[
 |A_j|+d\le(m-2d)+d=m-d\le m.
\]
The $j$th block therefore needs no basis reduction, so $h_j$ cannot reappear.

If $T'$ is not allowed, then $\mathbf A'$ is not top and
$j(\mathbf A')\in S'$.  But, since $j\notin S'$, we must have
$j(\mathbf A')\ne j(\mathbf A)$. As the selected block is determined by the set of low blocks, it must be the case that  $\High(\mathbf A')\ne\High(\mathbf A)$. By \Cref{cor:polynomial-multiplication},
\[
 w(T')
 \le w(x^{\mathbf A}h^{U'})+2\deg(P')-(2d+1)
 \le w(T)-1,
\]
where we used $\deg(P')\le d$ and the fact that the weight depends
only on $\mathbf A$, so $w(x^{\mathbf A}h^{U'})=w(x^{\mathbf A}h^{S}) = w(T)$.

Induction applies to every such term, while the
allowed terms already belong to $\mathcal R$. Thus
\Cref{eq:forbidden-rewrite} expresses $T$, on $E$, as a linear
combination of allowed generators.
\end{proof}

\section{Applications}
\label{sec:applications}

In \cref{sec:fourier}, we translate correlation bounds with XOR of
majorities into Fourier bounds. We use these in
\cref{subsec:fourier_prg} to obtain PRGs for Boolean-valued low-degree
polynomials and sparse polynomials over any arbitrary field $\F$.
In \cref{sec:AC0_parity}, we obtain PRGs for $\mathsf{AC}^0[q]$
circuits for every fixed prime power $q$. Finally, in \cref{sec:prf},
we construct pseudorandom functions that fool $\mathsf{AC}^0[2]$
oracle circuits.

\subsection{Fourier Bounds from Correlation with XOR of Majorities}\label{sec:fourier}

We show that our correlation bounds with XOR of majorities imply certain Fourier bounds in a generic way; in particular this resolves a conjectured Fourier bound for polynomials
\cite{ChattopadhyayGLLS21}. This has consequences for constructing PRGs for polynomials that we discuss in \cref{subsec:fourier_prg}.

For any Boolean function $f\colon\zo^n\to\zo$, define its sign
representation $\widetilde f\colon\{-1,1\}^n\to\{-1,1\}$ by
\(
\widetilde f(x)
= (-1)^{f((1-x_1)/2,\ldots,(1-x_n)/2)}.
\)
We write $\widehat f(S)$ for the real Fourier coefficients of
$\widetilde f$, so that
\[ 
\widetilde f(x)=\sum_{S\subseteq[n]}\widehat f(S)x_S,
\qquad
\widehat f(S)=\E_x[\widetilde f(x)x_S],
\]
where $x_S=\prod_{i\in S}x_i$ and $x$ is uniform on
$\{-1,1\}^n$. 

Define
\[
M_k(f)=\max_{x\in\{-1,1\}^n}
\left|\sum_{|S|=k}\widehat f(S)x_S\right|.
\]

For any class $\mathcal F$ of Boolean functions
$f\colon\zo^n\to\zo$, define
\[
M_k(\mathcal F)=\max_{f\in\mathcal F}M_k(f).
\]

Let $\XMajkl$ denote the XOR of Majority on $k$ disjoint input blocks of $\ell$ bits, for odd $\ell$. Our main result shows that $M_k$ bounds follow from correlation bounds with $\XMajkl$. Call a class of Boolean functions \emph{nice} if it is closed under input negations, restrictions, permuting input variables, and adding dummy variables.

\begin{restatable}{theorem}{fourierbound}\label{thm:fourier_from_corr} Let \(\mathcal{F}\) be a class of Boolean functions that is \emph{nice}. Suppose that for any \(k\in\{1,\ldots,n\}\), any odd \(\ell\ge1\), and any \(f\in\mathcal{F}\) on \(k\ell\) inputs, we have \(\operatorname{corr}(f,\XMajkl)\le(b/\sqrt{\ell})^k\).

Then, for all \(1\le k\le n\),
\[
M_k(\mathcal{F})\le(4b)^k.
\]
\end{restatable}

Before giving the proof, we observe that combining this result with \Cref{thm:total-degree} immediately yields the following Fourier bounds for polynomials, resolving the conjecture of~\cite{ChattopadhyayGLLS21}. 

\begin{corollary}\label{thm:fourier-mk}

For any field $\F$, let $\mathcal B^{\F}_{n,d}$ be the class
of Boolean functions on $\zo^n$ represented by polynomials over $\F$
of degree at most $d$. Then, for $1\le k\le n$,
\[
 M_k(\mathcal B^{\F}_{n,d})\le(16d)^k.
\]
\end{corollary}

\begin{proof}
Since Majority is versatile over every field $\F$, 
\Cref{main_cor} gives correlation at most
$(4d/\sqrt\ell)^k$ with $\XMajkl$. Apply \Cref{thm:fourier_from_corr}
with $b=4d$.
\end{proof}

\begin{proof}[Proof of \cref{thm:fourier_from_corr}]
Note that negating input coordinates multiplies $\widehat f(S)$ by the
corresponding product of signs. Thus it suffices to bound
$|\sum_{|S|=k}\widehat f(S)|$, for an arbitrary $f\in\mathcal F$. 

Consider a uniformly random partition $B=(B_1, \ldots, B_k)$ of $[n]$ (allowing empty parts).
We will show that \begin{equation}\label{eq:new_goal}
 \left|\sum_{S:\,|S\cap B_i|=1\ \forall i\in [k]}
          \widehat f(S)\right|\le \left(\frac{4b}{3}\right)^k
\end{equation}
for each choice of $B$.
Using the inequality $k^k/k! \le e^k$, this would complete the proof as 
\[
\E_B\left[\left|\sum_{S:\,|S\cap B_i|=1\ \forall i\in [k]}\widehat f(S)\right|\right] \ge 
\left|\E_B\left[\sum_{S:\,|S\cap B_i|=1\ \forall i\in [k]} \widehat f(S)\right]\right| = \left|\sum_{S:|S|=k}\widehat{f}(S) \cdot \frac{k!}{k^k}\right|.\]

Now, fix a partition $B_1, \ldots, B_k$ as above. It remains to prove \Cref{eq:new_goal}. 
Expanding the product gives
\[
\sum_{S:|S\cap B_i|=1\ \forall i\in[k]}
\widehat f(S)
=
\E_x\left[
\widetilde{f}(x)\prod_{i=1}^k\left(\sum_{j\in B_i}x_j\right)
\right].
\]
To apply our correlation bound, we will approximate each linear factor by a majority padded with many fresh variables.

Let $\ell = 2m+1 > n$ be odd and 
pad each part with fresh variables to length $\ell$,
and let $f$ ignore these additional variables. Think of $n$ as fixed and let $\ell$ grow to infinity.
On $\{\pm1\}$, majority is the function
$h_\ell(y)=\operatorname{sgn}(y_1+\cdots+y_\ell)$ and XOR is product.
Our assumed correlation bound implies that 
\begin{equation}\label{eq:padded-majority-correlation}
 \left|\E_x\left[\widetilde{f}(x)\prod_{i=1}^k h_\ell(x^{(i)})\right]\right|
 \le \left(\frac{b}{\sqrt\ell}\right)^k,
\end{equation}
where $x^{(i)}$ contains the variables of $B_i$ and its padding.

We recover the sum of coefficients that use one variable from each part by letting $\ell$ grow. We prove this as follows. 
The Fourier coefficients $\widehat{h_{\ell}}(T)$ of the majority function vanish on even $|T|$, while for odd $|T|=2j+1$,
\[
 \widehat h_\ell(T)
 =(-1)^j \frac{\binom{2m}{m}}{2^{2m}} \cdot \frac{\binom mj}{\binom{2m}{2j}},
\]
see, for example, \cite[Section~5.3]{ODonnell14} for a proof.

Expanding in the Fourier basis and dividing by $a_\ell^k$, where $a_{\ell} = \binom{2m}{m}2^{-2m}$ is the Fourier coefficient of the majority function on sets of size $1$, gives %
\[
 \frac{1}{a_{\ell}^k} \cdot \E_x\left[\widetilde{f}(x)\prod_{i=1}^k h_\ell(x^{(i)})\right]
 =
 \sum_{S\subseteq[n]}\widehat f(S)
 \prod_{i=1}^k
 \frac{\widehat h_\ell(S\cap B_i)}{a_{\ell}} .
\]
If $S$ misses a part, the corresponding factor is zero because
majority has mean zero. Otherwise, each factor is $1$ when
$|S\cap B_i|=1$, and tends to zero when $1<|S\cap B_i|\le n$.
Since the sum has only $2^n$ terms, taking $\ell\to\infty$ leaves
exactly those sets containing one variable from each part.
By \cref{eq:padded-majority-correlation}, we obtain 
\[
 \left|\sum_{S:\,|S\cap B_i|=1\ \forall i\in [k]}
          \widehat f(S)\right|\le \lim_{\ell \to \infty}\left(\frac{b}{a_{\ell} \sqrt{\ell}}\right)^k = \left(b\sqrt{\frac{\pi}{2}}\right)^k \le \left(\frac{4b}{3}\right)^k,\]
where we used $a_{\ell}\sqrt{\ell} \to \sqrt{2/\pi}$ as $\ell\to\infty$. %
\end{proof}

\subsection{PRG for Polynomials}
\label{subsec:fourier_prg}

We use the PRG from~\cite{ChattopadhyayGLLS21}
with the refinement
from~\cite[Theorem~5.10]{DBLP:journals/eccc/HatamiH23}.
If a family of Boolean functions $\mathcal F$ on $n$ variables
is closed under restrictions and negations of the input variables,
and satisfies $M_k(\mathcal F)\le b^k$ for some $b\ge1$ and
$k=\Theta(\log(b\log n/\varepsilon))$, then it has a PRG of
error $\varepsilon$ with seed length
$\widetilde O(b^2\log^2 n\,\log(1/\varepsilon))$.
Our Fourier bounds proved above imply that we can set $b=16d$
for Boolean-valued degree-$d$ polynomials over any  field $\F$.

\begin{theorem}\label{thm:mk-prg}
For every field $\F$, there is an explicit PRG that
$\varepsilon$-fools Boolean-valued degree-$d$ polynomials over $\F$
on $n$ Boolean inputs with seed length
$\widetilde O(d^2\log^2 n\,\log(1/\varepsilon))$.   
\end{theorem}

\begin{remark}As mentioned in the introduction, over $\F_2$ our construction improves
the dependence on $d$ from exponential to polynomial. Over an odd prime field $\F_p$, the construction of \cite{LMS13}
gives seed length
$O\!\left(\log n\cdot
\exp^{(2d+1)}((d/\varepsilon)^{C})\right)$,
where $\exp^{(t)}$ denotes the $t$-fold iterated exponential
and $C>0$ depends only on $p$.  
\end{remark}

\paragraph{Follow-up work.} For degree-$d$ polynomials over $\F_2$, Cohen~\cite{Coh26} refines the analysis
of the PRG in~\cite{ChattopadhyayGLLS21} by combining our new $M_k$ bounds with known $L_{1,k}$ bounds (\cite{DBLP:journals/toc/ChattopadhyayHH19}, see~\eqref{def:L1k} for its definition).
This gives seed lengths $\widetilde{O}(d^3 \log n \log(1/\varepsilon))$
and $O(d^4 \log n \log(1/\varepsilon))$, the latter achieving optimal
dependence on $n$ at the cost of a worse dependence on $d$.
\looseness=-1
 
\paragraph{PRG for sparse polynomials.} 
The result above can be used to construct PRGs for Boolean-valued sparse polynomials over any field. A polynomial is \emph{$S$-sparse} if its multilinear representation has at most $S$ nonzero monomials. Note that such polynomials may have monomials of arbitrary degree. Building on earlier work~\cite{LVW93,Viola-SYMAC0,DBLP:conf/approx/ServedioT18}, subsequent constructions over $\F_2$ achieved seed lengths $2^{O(\sqrt{\log S})}\log(1/\varepsilon)$~\cite{ServedioTan2022} and $2^{O(\sqrt{\log S})}+O(\log^{2.01}(1/\varepsilon))$~\cite{Kumar2025} for $S\ge n$. We reduce the sparsity dependence to polylogarithmic.

\begin{restatable}{corollary}{sparseprg}\label{thm:sparse-prg}
For every field $\F$, there is an explicit PRG that $\varepsilon$-fools all Boolean-valued $S$-sparse polynomials over $\F$ on $n$ Boolean inputs, with seed length
$\widetilde O\left(
  \log^2(S/\varepsilon)\log^2 n
  \log(1/\varepsilon)
\right)$.
\end{restatable}
\begin{proof}
Set $t=\lceil\log_{4/3}(6S/\varepsilon)\rceil$.
Let $Y$ be the output of the degree-$t$ PRG from
\Cref{thm:mk-prg}, with error $\varepsilon/6$.
Independently sample $Z$ from a $t$-wise independent distribution
on $\zo^n$ with uniform marginals, using $O(t\log n)$ random bits,
and put $D=Y\oplus Z$.

Conditioning on $Z$, bit flips are affine substitutions over $\F$,
so $D$ still $(\varepsilon/6)$-fools Boolean-valued degree-$t$
polynomials. Conditioning on $Y$ shows that $D$ is $t$-wise
independent. Take three independent copies
$D^{(1)},D^{(2)},D^{(3)}$ and output a distribution $X$ obtained by taking their coordinatewise majority: 
\(
 X_i=\mathsf{Maj}_3(D_i^{(1)},D_i^{(2)},D_i^{(3)}),
\)\footnote{Majority is
convenient because it is balanced and, for any choice of fixing two inputs uniformly, makes it constant
$0$ with probability $1/4$. Other choices also work, such as the
selector $g(a,b,c)=ab+(1-a)c$; its induced restrictions may also
negate inputs, which preserves degree and Boolean-valuedness.} where $\mathsf{Maj}_3$ denotes the  majority function on $3$ bits.

Fix a Boolean-valued $S$-sparse polynomial $P$.
We compare $X$ with the coordinatewise majority of three independent
uniform strings by replacing the three copies of $D$ one at a time.
At each replacement, the other two strings are independent and
$t$-wise uniform. Condition on their values $a,b$.
We claim that $P(\mathsf{Maj}_3(x,a,b))$ is a restriction of $P$: coordinate $i$ is set to $a_i$ if
$a_i=b_i$, and remains $x_i$ otherwise. In particular, this
restriction remains Boolean-valued.

A monomial of degree greater than $t$ is set to zero if
$(a_i,b_i)=(0,0)$ for any coordinate in its support.
Considering any $t$ such coordinates, the probability that the
monomial is not set to zero is at most $(3/4)^t$.
A union bound therefore shows that, with probability at least
$1-(3/4)^t\cdot S$ over $a,b$, the restricted polynomial has degree
at most $t$.
Consequently, each replacement changes the acceptance probability
by at most $\varepsilon/6+S(3/4)^t$.

The coordinatewise majority of three independent uniform strings
is uniform. Summing the errors of the three replacements gives
\[
 \left|\Pr[P(X)=1]-\Pr[P(U_n)=1]\right|
 \le 3\left(\varepsilon/6+S(3/4)^t\right)
 \le \varepsilon.
\]
Substituting $t=O(\log(S/\varepsilon))$ into
\Cref{thm:mk-prg}, adding the $O(t\log n)$ random bits for each
copy of $Z$, and taking three independent copies gives the
claimed seed length.
\end{proof}

\subsection{\texorpdfstring{PRG for $\mathsf{AC}^0[q]$}{PRGs for AC0[q]}}\label{sec:AC0_parity}

For a fixed prime power $q$, the class $\mathsf{AC}^0[q]$ consists of Boolean functions computed by families of polynomial-size, constant-depth circuits with unbounded-fan-in AND, OR, NOT gates and $\mathrm{MOD}_q$ gates. A $\mathrm{MOD}_q$ gate outputs $0$ if the sum of its input bits is congruent to $0$ modulo $q$, and $1$ otherwise. 

We derive the following PRGs for $\mathsf{AC}^0[q]$, for every fixed prime power $q$, from our results on fooling Boolean-valued polynomials. 

\begin{theorem}\label{thm:circuit-prgs}
Fix a prime power $q=p^r$ and a depth $\Delta\ge1$. There is an explicit PRG that
$\varepsilon$-fools size-$s$, depth-$\Delta$
$\mathsf{AC}^0[q]$ circuits on $n$ inputs  with seed length
$\widetilde O(\log^{2\Delta}(s)\cdot\log^3(1/\varepsilon))$.
\end{theorem}

\begin{proof}
Fix such a circuit $C$. Each $\mathrm{MOD}_q$ gate has an exact polynomial
of degree at most $q-1$ over $\F_p$. Indeed, if $e_t$ denotes the
$t$-th elementary symmetric polynomial in the gate's input bits, then
\[
 \mathrm{MOD}_{p^r}(z)
 =1-\prod_{j=0}^{r-1}\left(1-e_{p^j}(z)^{p-1}\right).
\]
By Lucas's theorem \cite{lucas1878theorie}, %
$e_{p^j}(z)=\binom{|z|}{p^j}\bmod p$ is the
$j$-th base-$p$ digit of $|z|$, and the polynomial on the right-hand side  has degree
at most $(p-1)\sum_{j=0}^{r-1}p^j=q-1$.
Thus the polynomial approximation method of
Razborov and Smolensky~\cite{Raz87,Smo87}, with the same error-reduction
argument as in \cite{DBLP:journals/toc/Kopparty018},
gives a distribution of polynomials $P$ over $\F_p$ of degree at most
$d=O_{q,\Delta}(\log^{\Delta-1}(s)\log(1/\varepsilon))$ such that, for every
input $x$, $\Pr_P[P(x)\ne C(x)]\le\varepsilon/3$.
Replace each $P$ by $Q=P^{p-1}$. Then $Q$ is Boolean-valued, has degree
at most $(p-1)d$, and agrees with $C$ whenever $P$ does.
By \cref{thm:mk-prg}, there is a PRG that $(\varepsilon/3)$-fools
all such $Q$ with seed length
$\widetilde O(d^2\log^2 n\,\log(1/\varepsilon))$.
Averaging over $Q$ and applying the triangle inequality, its error for $C$
is at most $\varepsilon/3$ for each of the two approximation errors and
$\varepsilon/3$ for fooling $Q$, hence at most $\varepsilon$.
\end{proof}

\subsection{Pseudorandom Functions}\label{sec:prf}
\emph{Pseudorandom functions} are distributions on functions that can be generated from a short seed yet cannot be distinguished from a truly random function by efficient circuits equipped with oracle gates that provide evaluations of the function.  They were introduced by Goldreich, Goldwasser, and Micali~\cite{GGM86},
 who constructed them from cryptographic assumptions.  We construct unconditional pseudorandom functions (PRFs) with
polylogarithmic key lengths that fool $\mathsf{AC}^0[2]$
oracle circuits.  To our knowledge, such a construction was
not known even for ordinary $\mathsf{AC}^0$ oracle circuits.
The construction critically relies on our PRG for circuits with
parity gates.

We equip an $\mathsf{AC}^0[2]$ circuit $C$ with oracle gates that evaluate a function $f\colon\zo^n\to\zo$. 
We write $C^f$ for the output of the circuit. Oracle gates count toward both size and
depth; their inputs may depend on earlier oracle answers.

\begin{theorem}\label{thm:prf}
There is an explicit family of functions $f_k\colon\zo^n\to\zo$, indexed by a key
$k\in\zo^r$ with length $r=\log^{O(\Delta)}(ns/\varepsilon)$, such that for every
size-$s$, depth-$\Delta$ $\mathsf{AC}^0[2]$ oracle circuit $C$ we have
\[
 \left|\Pr_k[C^{f_k}=1]-\Pr_R[C^R=1]\right|\le\varepsilon,
\]
where $k$ is uniformly random and $R\colon\zo^n\to\zo$ is a uniformly random function.
Given a key $k$ and an input $x\in\zo^n$, the value $f_k(x)$ is computable in time polynomial in $n+r$.
\end{theorem}

The idea is to use \cref{thm:circuit-prgs} to generate a
pseudorandom table of polynomially many bits.  The pseudorandom function then hashes each query to a position in this table and returns the corresponding bit.
The analysis proceeds as follows.  With a uniform table, the answers are those of a random function unless two distinct queries collide. For a fixed hash, replacing
the oracle calls with hashing and table lookups gives an ordinary
$\mathsf{AC}^0[2]$ circuit on the table bits, so our PRG fools it.
Note that hashing introduces parity gates in this circuit; thus the
construction uses the $\mathsf{AC}^0[2]$ PRG even for ordinary $\mathsf{AC}^0$
oracle distinguishers.

\paragraph{Construction.}
Let $A$ be a sufficiently large absolute constant, chosen below,
and let $Q$ be the smallest power of two $\ge Ans^2/\varepsilon$.
Identify $\F_Q$ with $\zo^{\log Q}$ using a fixed $\F_2$-basis of $\F_Q$,
and index $Q$-bit tables by
$\mathbb F_Q$. Take the generator $G\colon\zo^u\to\zo^Q$ from
\cref{thm:circuit-prgs}, that $\eps/2$-fools circuits of size at most 
$As(Q+n)\log Q$, and depth at most $A\Delta$, which has seed length $u=\log^{O(\Delta)}(ns/\eps)$.
Choose $a\in\mathbb F_Q$ and $z\in\zo^u$ independently and uniformly at random. Set
\[
 H_a(x) \coloneq \sum_{i=1}^n x_i a^{i-1},\qquad
 f_{a,z}(x)\coloneq G(z)_{H_a(x)}.
\]
The key is $k=(a,z)$, so its length is
$r=\log Q+u=\log^{O(\Delta)}(ns/\varepsilon)$.

\begin{proof}
We write $f_{a,z}$ as $G(z)\circ H_a$.  Our goal is to show that any size-$s$, depth-$\Delta$ circuit cannot distinguish $R$ from $G(z)\circ H_a$.
To this end, we introduce a hybrid $T \circ H_a$, where $T$ is an independent uniform $Q$-bit table, and show that no such circuit can distinguish $R$ from $T \circ H_a$ or $T \circ H_a$ from $G(z)\circ H_a$, except with small advantage.

\paragraph{Random function versus hashed uniform table.}
Fix an execution of an oracle circuit, and consider any transcript $\tau$ of $m$ query-answer pairs that records only the first occurrence of each distinct query.
We have $m \le s$ as the circuit can only make at most $s$ oracle queries.  Note that this transcript determines the queries themselves, as well as the output of the circuit.
The transcript occurs with probability $2^{-m}$ under $R$.  Under the hashed uniform table, the transcript occurs with probability at least $2^{-m}$ multiplied by the probability $\pi$ that no two distinct queries collide under the hash function $H_a$.
To lower bound $\pi$, note that for distinct $x,y$, the polynomial $H_a(x)-H_a(y)$ is nonzero and has
degree at most $n-1$, so its collision probability is at most $(n-1)/Q$.
By a union bound we have
\[
 \pi \ge 1-\binom{s}{2}\frac{n-1}{Q}\ge 1-\frac{\varepsilon}{A}.
\]
Summing over the transcripts that make the circuit output $1$, and similarly over the transcripts that make the circuit output $0$, shows that the distinguishing advantage between the two distributions on functions is at most $\varepsilon/A$.

\paragraph{Uniform table versus PRG table.}
Fix $a$. Since $H_a$ is $\F_2$-linear, every coordinate of $H_a(x)$ with respect to the fixed $\F_2$-basis is a parity of the input bits.
Replace each oracle gate by this parity layer followed by a table look-up.  Each replacement has constant depth and size
$O((Q+n)\log Q)$. The resulting circuit on the table bits belongs to
$\mathsf{AC}^0[2]$, has size $O(s(Q+n)\log Q)$, and depth $O(\Delta)$.
Choose $A$ large enough so that the circuit is covered by the PRG
parameters and the collision error is at most $\varepsilon/2$.
By pseudorandomness of the PRG, the distinguishing advantage between $T\circ H_a$ and $G(z)\circ H_a$ is at most $\eps/2$.
Averaging over all $a$ and applying the triangle inequality completes the pseudorandomness guarantee.

It remains to show that $f_{a,z}(x)$ can be computed efficiently.  To compute $f_{a,z}(x)$, we compute $H_a(x)$ and then one specified
output bit of $G(z)$. We use the following implementation of the
polarizing-walk generator
\cite{DBLP:journals/toc/ChattopadhyayHH19,ChattopadhyayGLLS21}. (Alternatively, a direct implementation of Nisan's generator \cite{Nis91} gives the desired efficiency.)
Its seed specifies $N$ independent samples
$Z^{(1)},\ldots,Z^{(N)}\in\{-1,1\}^Q$, each drawn from a
$\kappa$-wise independent distribution. To compute output coordinate $i$,
start with $y_0=0$ and apply
\[
 y_j=y_{j-1}+(1-|y_{j-1}|)cZ_i^{(j)}
 \qquad (j=1,\ldots,N),
\]
where $c=2^{-\lambda}$ for a positive integer $\lambda$.
The output bit is $(1-\operatorname{sign}(y_N))/2$, with
$\operatorname{sign}(0)=1$. For our parameters, the cited analysis
allows choices with
$N \kappa\log Q=\log^{O(\Delta)}(ns/\varepsilon)$ and
$\lambda=O(\log N)$.\footnote{For this application, we do not need the randomized-rounding refinement
of~\cite[Theorem~5.10]{DBLP:journals/eccc/HatamiH23} used to obtain the parameters in \cref{thm:mk-prg}.
Since $Q\ge 1/\varepsilon$, we have
$\log(Q/\varepsilon)=O(\log Q)$, so the original walk with deterministic
sign rounding already achieves the claimed polylogarithmic key length.}
Standard bounded-independence constructions compute any specified
$Z_i^{(j)}$ from its seed and $i$ in time polynomial in $\kappa$ and
$\log Q$, using finite-field polynomial evaluation (see, e.g., the proof of \cite[Theorem~2.2] {DBLP:journals/eccc/HatamiH23}).
We can therefore compute $G(z)_i$ by following the above 
recurrence relation. Every intermediate value has denominator dividing
$2^{\lambda N}$, so the computation uses exact arithmetic on
$O(\lambda N)$-bit integers and takes time polynomial in the seed
length. Together with evaluating $H_a(x)$ using $O(n)$ operations
in $\F_Q$, this gives evaluation time polynomial in $n+r$.
\end{proof}

\section{Open Problems} \label{sec:open}
Our results leave several natural questions.

\paragraph{Better correlation bounds.}
Our construction gives correlation at most $2^{-\Omega(n/d^2)}$ with degree-$d$ polynomials over $\F_2$.
A first target towards improving this would be to obtain an explicit function with correlation $o(1)$ against every degree-$d$ polynomial for some $d=\omega(\sqrt n)$.
More generally, can we obtain correlation $2^{-\Omega(n/d)}$ for an explicit function when $d=o(n)$? Since our bound for the XOR of majorities is tight up to constants, this would require a different hard function.
An even more ambitious bound would imply new general Boolean circuit lower bounds: for some constant $\delta>0$, an explicit function with correlation at most $2^{-\omega(n/\log\log n)}$ against all degree-$n^\delta$ polynomials would require superlinear-size bounded-fan-in circuits of depth $O(\log n)$~\cite{Val77,GKW}.

\paragraph{Improved PRGs.}
\Cref{thm:mk-prg} gives seed length $d^2\cdot\polylog(n/\eps)$ for degree-$d$ polynomials.
A natural goal would be to reduce the dependence on $d$ from quadratic to linear. We note that non-constructively it is possible to achieve seed length $O(d\log n+\log(1/\eps))$, which is tight up to constants \cite{BoV-gen,AlonBK08}.

\paragraph{Stronger Fourier bounds.} We bound $M_k(f)$, which takes the absolute value after summing
the signed level-$k$ Fourier coefficients and thus allows cancellation.
It remains open whether the sum of their absolute values is also
bounded by $O(d)^k$ for every degree-$d$ polynomial.
For a Boolean function $f:\zo^n \to \zo$,  define
\begin{equation}\label{def:L1k}
 L_{1,k}(f)\coloneqq\sum_{|S|=k}\left|\widehat f(S)\right|.
\end{equation}
 For any $f \in \mathcal B^{\F_2}_{n,d}$, the  bound $L_{1,k}(f)\le (Cd)^k$, for an absolute constant $C$,
was conjectured for $k=2$ in~\cite{DBLP:conf/innovations/ChattopadhyayHL19}
and for all $k$ in~\cite{BIJLSV}. The conjecture was proved for symmetric polynomials
in~\cite{BIJLSV}. For general polynomials, the best known bound is
$L_{1,k}(f) \le(k2^{3d})^k$~\cite{DBLP:journals/toc/ChattopadhyayHH19}.
A confirmation of this conjecture would yield a PRG for degree-$d$ polynomials over $\F_2$ with seed length
$\widetilde O(d^2 \log n \log(1/\eps))$, saving a factor of $\log n$ over the bound in
\Cref{thm:mk-prg}.

\bibliographystyle{alpha}
\bibliography{OmniBib}

\end{document}